\documentclass[journal]{IEEEtran}
\usepackage{amsmath,graphicx,xcolor,amsfonts,tabularx}
\usepackage{multirow,cancel,booktabs,amsthm}
\newtheorem{Lem}{Lemma} 
\newtheorem{Thm}{Theorem}

\newcommand{\st}{%
\mathrel{\ooalign{$\ni$\cr\kern-1pt$-$\kern-6.5pt$-$}}}
\renewcommand{\baselinestretch}{1.0}

\usepackage{algorithm}
\usepackage{algpseudocode}

\ifCLASSINFOpdf

\else

\fi

\begin{document}

\title{Compressive Sensing of Large-Scale Images:\\An Assumption-Free Approach}

\author{Wei-Jie~Liang,~
        Gang-Xuan~Lin,
        and~Chun-Shien~Lu
\thanks{W.-J. Liang and G.-X. Lin are with Department of Mathematics, National Cheng-Kung University, Tainan, Taiwan, ROC}
\thanks{C.-S. Lu is with Institute of Information Science, Academia Sinica, Taipei, Taiwan, ROC}
}

\maketitle

\begin{abstract}
Cost-efficient compressive sensing of big media data with fast reconstructed high-quality results is very challenging.
In this paper, we propose a new large-scale image compressive sensing method, composed of operator-based strategy in the context of fixed point continuation method and weighted LASSO with tree structure sparsity pattern. The main characteristic of our method is free from any assumptions and restrictions. The feasibility of our method is verified via simulations and comparisons with state-of-the-art algorithms.%
\end{abstract}


\IEEEpeerreviewmaketitle

\section{Introduction}
\IEEEPARstart
{C}{ompressive} sensing (CS) \cite{Baraniuk07, CRT06, Donoho06} of sparse signals in achieving simultaneous data acquisition and compression has been extensively studied in the literature.
In the context of CS, we usually let $u$ denote a $k$-sparse signal of length $n$ to be sensed, let $\Phi$ of dimensionality $m \times n$ represent a sampling matrix, and let $y$ be the measurement of length $m$, where $k< m < n$ and $0<\frac{m}{n}<1$ is defined as the measurement rate.
At the encoder, random projection, defined as:
\begin{equation}
y=\Phi u,
\label{Eq: Random Projection-1D}
\end{equation}
is conducted on the original signal $u$ via $\Phi$ to obtain the measurement vector $y$.
At the decoder, $u$ can be recovered based on its sparsity by means of convex optimization or greedy algorithms.

Compressive sensing has been widely studied for $1$D signals and $2$D images with reasonable sizes.
However, efficient compressive sensing of big media data without making any assumptions has not been found in the literature.
We can foresee that when signal length $n$ is large, almost all existing CS algorithms will run out of memory and/or the computations will be overloaded.
To address the need of compressive sensing of big media data in both efficient sensing and recovery aspects, we study a new scheme for big media signals.
Please note that the storage consumption for a random matrix is unacceptable because it needs $2$G memory to store a random matrix with (small) size $n=\left(128\times128\right)^2$.


\subsection{Related Work}
In \cite{Duarte12}, Duarte and Baraniuk introduce Kronecker product to model multidimensional compressive sensing of signals and propose a Kronecker compressive sensing (KCS) method.
They prove that the Kronecker sensing matrix and Kronecker dictionary possess mutual incoherence property 
(MIP) and restricted isometric property (RIP).
Nevertheless, the practical use of KCS is greatly prohibited because the vectorization of multidimensional signals and the use of joint sensing involve a (very) large Kronecker product-based sensing matrix.

In \cite{Sidiropoulos12}, a multiway compressive sensing (MWCS) method for sparse and low-rank tensors is proposed. 
Although MWCS achieves more efficient reconstruction, 
its performance relies heavily on tensor rank estimation, which is NP-hard.
A generalized tensor compressive sensing (GTCS) method for higher-order tensors has been proposed in \cite{Li13}.
GTCS is demonstrated to be comparable to KCS in recovery accuracy and be greatly faster than KCS in recovery speed.

While the previous studies consider some tensor operations like Kronecker product, CP decomposition, and Tucker model within the framework of compressive sensing, the sparsity pattern inherent in the big media data/tensor (like $2$D image and $3$D video) has not been fully explored.
Recently, Caiafa and Cichocki \cite{Caiafa13} exploit Kronecker product and block sparsity to develop a so-called N-BOMP (N-way block OMP).
However, we find, as also indicated in subsection 7.2.1 of \cite{Caiafa13}, that for a $2$D image it is pre-processed in advance to possess perfect block sparsity pattern in that the important/significant coefficients in some transform domain fall within the specified block sparsity pattern while other insignificant coefficients are entirely removed.
Under the situation, N-BOMP is able to obtain reconstruction quality far better than the existing tensor CS algorithms.

Recently, Caiafa and Cichocki \cite{Caiafa14} present a fast tensor compressive sensing method, which no longer assumes certain sparsity pattern and does not involve iterations, making it suitable for large-scale problems.
However, it assumes that the signal to be sensed and recovered has low multilinear-rank, leading to redundant sensing, which means that under the same measurement rate the reconstructed quality is (remarkably) lower than other CS algorithms.

In \cite{Lu14}, we previously propose the use of tree-structure sparsity pattern (TSSP) in tensor compressive sensing.
TSSP can help to fast find significant wavelet coefficients.
Its weakness is that there does not exist a fast recovery algorithm that can exploit TSSP.
In this paper, we conquer this problem.

Hale {\em et al.} \cite{W.Yin08} derive the optimality conditions of convex optimization problem
as a fixed-point equation.
Later, Wen {\em et al.} \cite{W.Yin10}\cite{W.Yin12} provide a fast algorithm, called FPC\_AS, to solve the convex optimization problem.
Specifically, their sensing matrix for convex optimization problem is chosed as a partial DCT matrix, whereas
the fast algorithm is based on active-set strategy, which can solve the fixed-point equation in large-scale problems.


\subsection{Motivation and Our Contributions}
As indicated in Eq. (\ref{Eq: Random Projection-1D}), when the signal length becomes large enough, the storage and computation of using the sensing matrix become an obstacle.
To deal with sensing and recovery of large-sized images, we do not follow the convention of dividing a large image into several small blocks \cite{Gan07, Mun09, Mun11, YPL13}, wherein each small block can be separately processed.
This will incur either blocky effects in recovery or required calibration in block-based sensing.

Based on the above concerns, we explore the sensing strategy \cite{Do12} and recovery strategy \cite{W.Yin08, W.Yin10, W.Yin12, Milzarek14} that can be operated in operators to speed up computation and save storage without needing any assumptions/restrictions.
In this paper, the recovery strategy is based on exploiting fixed-point equation \cite{W.Yin08} to solve the convex optimization problem instead of greedy algorithms that rely on matrix computation, which is prohibited in the context of big data compressive sensing.
Specifically, the recovery method is based on the FPC\_AS algorithm \cite{W.Yin10, W.Yin12} that solves the fixed-point equation.
To recover the large-scale data,
the authors in \cite{W.Yin08, W.Yin10, W.Yin12} propose an active-set algorithm and Milzarek {\em et al.} \cite{Milzarek14} propose a globalized semismooth Newton method to solve large-scale $l_1$-norm optimization problems,
where partial DCT matrix is adopted as the sensing matrix for fast sensing.
However, the signal to be sensed in \cite{Milzarek14} is assumed to be sparse in the time domain (that is the sparsifying basis is the identity matrix).

In this paper, we propose to use a kind of random Gaussian-like matrix, called Structurally Random Matrix (SRM) \cite{Do12}, as the sensing matrix, and wavelet as the sparsifying basis to deal with large-scale images.
Our choice can satisfy RIP and MIP in CS.
Basically, our method can be viewed as an extension of \cite{Milzarek14} to compressive sensing of large-scale signals that are not sparse themselves but such extension is not trivial at all.


In addition to the above, for sparse recovery of big images, the sparsity pattern plays an important role.
With an eye on the natural characteristic of tree-structure relationship among wavelet coefficients that are popularly used to represent media data, we propose to explore tree-structure sparsity pattern (TSSP) in big data compressive sensing.
When TSSP is considered (with the way different from \cite{Lu14}), our strategy is to assign smaller weights to the wavelet coefficients at the lower frequencies and larger weights to those at the higher frequencies under the framework of convex optimization.
Specifically, we explore a weighted-LASSO algorithm for sparse recovery of big images from sensed measurements.


\section{BIG IMAGE COMPRESSIVE SENSING}
\label{sec:pagestyle}

In this section, we describe the proposed method, wherein fixed point strategy is adopted to solve the Lasso problem and the step size is controlled by quasi-Armijo rule.

\subsection{System Model}
Let $U \in \mathbb{R}^{N \times N}$ be the $2$-dimensional image.
It can be sparsely represented via certain dictionaries $\Psi_1$ and $\Psi_2$ as:
\begin{equation}
	U = \Psi_1 X \Psi_2^T,
\label{sparse}
\end{equation}
where $X$ is sparse with respect to $\Psi_1$ and $\Psi_2$.
Here, we reshape $2$D signal to $1$D vector.
Based on Eq. (\ref{Eq: Random Projection-1D}), we have
\begin{equation}
	y = \Phi u = \Phi \Psi x = Ax,
\label{CS}
\end{equation}
where $A = \Phi \Psi\in\mathbb{R}^{m\times n}$, $n=N^2$, sparsifying basis $\Psi = \Psi_2 \otimes \Psi_1\in\mathbb{R}^{n\times n}$, $u = vec(U)$, and $x = vec(X)\in\mathbb{R}^n$.

The main problem is to reconstruct the vector $x$ from the measurement $y$ in Eq. (\ref{CS}).
A well-established approach for the reconstruction
is an optimization method which we call the
LASSO (or penalized least-squares) problem:
\begin{equation}
\displaystyle\min_{x\in \mathbb{R}^n}\lambda\left\|x\right\|_1 + \frac{1}{2}\left\| y-Ax\right\|_2^2, 
\label{Lasso}
\end{equation}
where $\lambda>0$ specifies the penalty of sparse level of $x$.
To ease discussion later, we set $\mathcal{F}(x) = \lambda\left\|x\right\|_1 + \frac{1}{2}\left\| y-Ax\right\|_2^2$.

\subsection{Sensing Matrix Design}\label{Sensing Matrix Design}
The common use of random Gaussian matrix as the sensing matrix leads to problems of storage and computation cost.
Although storage consumption can be overcome by using a seed to generate random Gaussian, it still encounters high computational cost.
Nguyen {\em et al.} in \cite{Do12} propose a framework, called Structurally Random Matrix (SRM), with:
\begin{equation}
	\Phi = DFR,
\label{DFR}
\end{equation}
where $D \in\mathbb{R}^{m \times n}$ is a sampling matrix, $F \in\mathbb{R}^{n \times n}$ is an orthonormal matrix, and $R \in\mathbb{R}^{n \times n}$ is a uniform random permutation matrix.
Since the distributions between a random Gaussian matrix and SRM's $\Phi$ are verified to be similar, we choose Eq. (\ref{DFR}) as the sensing matrix for our use.

In this paper, we set $F$ to the Discrete Cosine Transform (DCT) due to its fast computation and cost-effectiveness.

\subsection{Fixed Point Method with Quasi-Armijo Rule}
Due to the convexity of the function $\mathcal{F}(x)$ in Eq. (\ref{Lasso}), the global minimum solution
is exactly the critical point of $\mathcal{F}(x)$.
In \cite{Milzarek14}, the authors derive that the critical point of $\mathcal{F}(x)$ exactly belongs to the solution set of fixed point equation as
\begin{eqnarray}\label{J}
\mathcal{J}=\left\{x\bigg|x = S_{\tau\lambda}(G_{\tau}(x))\right\},
\end{eqnarray}
where
$\tau>0$ is arbitrarily fixed,
\begin{eqnarray}
\label{fixed point}
\left\{\begin{array}{c}
G_{\tau}(x) = x-\tau A^T(Ax-y),\\
S_{\tau\lambda}(x) = x-\mathcal{P}_{\left[-\tau\lambda,\tau\lambda\right]}(x),
\end{array}\right.
\end{eqnarray}
and
$\mathcal{P}_{\left[-c,c\right]}(t)=\min\left\{\max\left\{-c,t\right\},c\right\}$ is the projection
onto the interval $\left[-c,c\right]$.
Then the fixed point iteration is defined as:
\begin{equation}
	x^{k+1} = S_{\tau\lambda}(G_{\tau}(x^k)) \quad \mbox{with} \quad \tau >0.
\label{fixed point iter}
\end{equation}
Let $x^k$ be the current iterate and let $d^k=x^{k+1}-x^k$ be a direction that is generated by Eq. (\ref{fixed point iter}).
Then we can calculate $x^{k+1}$ by $x^k+\sigma_kd^k$, where $\sigma_k$ is the step size controlled by a quasi-Armijo rule.
The algorithm is depicted in Algorithm \ref{alg1}, where we extend it to deal with big images that they themselves are not sparse.

\begin{algorithm}[!hbpt]
  \caption{ \cite{Milzarek14}
Fixed point method with quasi-Armijo rule.}
  \label{alg1}
  \begin{algorithmic}[1]
\raggedright
    \Require
      The initial iterative point, $x^0 = 0 \in \mathbb{R}^n$;
	  The shrinkage parameter, $\tau>0$;
	  The quasi Armijo's step size parameter, $\beta\in(0,1)$;
      A constant, $\gamma\in(0,1)$;
	  The weighted parameter of LASSO, $\lambda>0$;
	  An initial iterative step, $k=0$;

    \Ensure
      The $k^{\footnotesize \mbox{th}}$ iterative point, $x^k$;

    \State Calculate the direction: $d^k = S_{\tau\lambda}(G_{\tau}(x^k))-x^k$;
    \label{alg1:direction}

	\While{$d^k\neq0$}

	\State Calculate $\triangle_k=\left(d^k\right)^{T}\left(A^T\left(Ax-y\right)\right)
+\lambda\left(\left\|S_{\tau\lambda}\left(G_{\tau}\left(x^k\right)\right)\right\|_1-\left\|x^k\right\|_1\right)$;%
	\label{alg1:delta}

    \State Choose a maximal quasi-Armijo step size
$\sigma_k\in\left\{1, \beta, \beta^2, \ldots\right\}$
such that
$\mathcal{F}(x^k+\sigma_kd^k)-\mathcal{F}(x^k)\leq\sigma_k\gamma\triangle_k$;
    \label{alg1:quasi Armijo}

    \State $x^{k+1} = x^k + \sigma_kd^k$;
    \label{alg1:iterative}

	\EndWhile

	\State \textbf{return} $x^k$;
	\label{alg1:return}
  \end{algorithmic}
\end{algorithm}


\subsection{Tree Structure Sparsity in Convex Optimization}
\label{Tree Structure Sparsity in Convex Optimization}

For improving the quality of reconstruction, we refer to an iterative reweighted $l_1$-norm minimization (IRWL1), which is proposed in \cite{CWB08}.
For a given diagonal weighted matrix $W\in\mathbb{R}^{n\times n}$, the convex optimization problem in
Eq. (\ref{Lasso}) can be relaxed with $\hat{x}=Wx$ and $\hat{A} = \Phi \Psi W^{-1}$ as:
\begin{equation}
	\displaystyle\min_{\hat{x} \in \mathbb{R}^n}
	\lambda \left\|\hat{x}\right\|_1 + \frac{1}{2}\left\|\hat{A} \hat{x} - y\right\|_2^2.
\label{Lasso_weighted}
\end{equation}
The main difference between our model in Eq. (\ref{Lasso_weighted}) and IRWL1
is that the weighted matrix $W$ is determined by tree structure sparsity pattern (TSSP).

More specifically, TSSP is yielded by separating the wavelet coefficient for 2D image into support and not-support sets.
We adopt the wavelet, as provided in the source code of \cite{Caiafa13}, as the dictionary $\Psi$ and apply it to an image with $S$ levels to obtain the subbands \footnotesize$\left\{LL_S, LH_S, HL_S, HH_S, LH_{S-1},\dots,HH_1\right\}$\normalsize, where $L$ and $H$ denote low and high frequencies, respectively.
The resultant wavelet coefficients are weighted according to which levels of subbands they are located in.

Algorithm \ref{alg2} describes how to solve Eq. (\ref{Lasso_weighted}).
First, since the coefficients in $LL_S$ are significant, the corresponding indices in $W$ are all reserved and set to $0.1$ (as in initialization part of Algorithm \ref{alg2}), and other diagonal elements are set to $1$.
We solve Eq. (\ref{Lasso_weighted}) (Step \ref{alg2:first iteration} in Algorithm \ref{alg2}) to yield the initial solution.

Second, we check the entries from the subbands $LH_S$, $HL_S$ and $HH_S$ that could be the roots of evolving trees
(Steps \ref{alg2:truncated}, \ref{alg2:residue},
and \ref{alg2:correlation} in Algorithm \ref{alg2}).
They will be put in the queue $Q$ if they are supports (corresponding coefficients are large enough (Step \ref{alg2:0.2} in Algorithm \ref{alg2})).
For each element in $Q$, if it is checked to be a support, its children will be put in $Q$
(Steps \ref{alg2:0.2} and \ref{alg2:sub1} in Algorithm \ref{alg2}).
This process is repeated until $Q$ is empty to complete the generation of TSSP.

To construct $W$,
its diagonal entries are decided by TSSP.
We expect that the wavelet coefficients, solved by
Eq. (\ref{Lasso_weighted}), are located on the tree decided by TSSP.
Since the decision variable with small weight will derive the
large wavelet coefficient, together with the fact that
the energy is decreasing from level $S$ to level 1,
the weights are empirically set to:
\begin{equation}
\label{W}
diag(W)_i=0.1(S-s+1)
\mbox{\ \  if \ \ }
i\in I_s,
\end{equation}
where $I_s$ is the subset of $LH_s\cup HL_s\cup HH_s$ and is decided by TSSP.


Finally, we solve Eq. (\ref{Lasso_weighted}) with weighted matrix $W$ in Eq. (\ref{W}) to obtain the final solution (Step \ref{alg2:third iteration} in Algorithm \ref{alg2}).

\begin{algorithm}[!hbpt]
  \caption{ IRWL1 with TSSP.}
  \label{alg2}
  \begin{algorithmic}[1]
\raggedright
    \Require
	  The initial weighted iterative point, $\hat{x}=0\in\mathbb{R}^n$;
	  The initial truncated iterative point, $\tilde{x}=0\in\mathbb{R}^n$;
	  The weighted matrix, $W=I_{n\times n}$;
	  The index set of $LL_S$, $I_{S+1}$;
	  The empty sets, $I_s$, $1\leq s\leq S$;
	  The percentage, $p\%$;
	  A tolerance, $\epsilon>0$;

    \Ensure
      The final output, $x^*$;

	\State Set the weighted of $LL_S$ as $diag(W)\big|_{I_{S+1}}=0.1$;
	\label{alg2:0.1}

    \State Solve Eq. (\ref{Lasso_weighted}) by Alg. \ref{alg1} with weighted matrix $W$ to obtain the solution $x^*$;
    \label{alg2:first iteration}

    \State Set the truncated variable $\tilde{x}\big|_{I_{S+1}}=x^*\big|_{I_{S+1}}$;
    \label{alg2:truncated}

	\State Calculate residue $r = y - \hat{A}\tilde{x}$;
	\label{alg2:residue}

    \State Calculate correlation $c_i = \left|\hat{A_i}^Tr\right|$, $i=1,2,\ldots,n$;
    \label{alg2:correlation}

    \State Let $I_S$ collects the indices with the first $p\%$ largest correlations in $LH_S\cup HL_S\cup HH_S$.
Set $diag(W)\big|_{I_S} = 0.2$;
    \label{alg2:0.2}

	\For{$s=S-1, S-2, \ldots, 1$}

	\State Construct $I_s$ as the children of $I_{s+1}$, set $diag(W)\big|_{I_s}$ according to Eq. (\ref{W});
	\label{alg2:sub1}

	\State Solve Eq. (\ref{Lasso_weighted}) by Alg. 1
to obtain the solution $x^*$;
	\label{alg2:second iteration}

	\For{$j\in I_s$}

	\If{$\left|x^*_j\right|<\epsilon$}

	\State remove index $j$ and set $diag(W)_j=1$;
	\label{alg2:sub3}

	\EndIf

	\EndFor

	\EndFor

	\State Solve Eq. (\ref{Lasso_weighted}) by Alg. 1 with final
weighted matrix $W$ and output the solution $x^*$;
	\label{alg2:third iteration}

	\State \textbf{return};
	\label{alg1:return}
  \end{algorithmic}
\end{algorithm}

\subsection{Memory Cost and Computational Complexity}
The main computational cost of solving Eq. (\ref{Lasso_weighted}) comes from the computation of matrix-vector multiplications:
$$\hat{A}\hat{x}=\Phi\Psi W^{-1}\hat{x}=DFR\mathcal{W}W^{-1}\hat{x},$$
including $G_{\tau}(\hat{x})=\hat{x}-\tau \hat{A}^T\left(\hat{A}\hat{x}-y\right)$
and $\mathcal{F}\left(\hat{x}^k\right)=\lambda\left\|\hat{x}\right\|_1
+\frac{1}{2}\left\|y-\hat{A}\hat{x}\right\|_2^2$,
where $D$, $F$, and $R$ are defined in Sec. \ref{Sensing Matrix Design},
$\mathcal{W}$ is an inverse wavelet transform, and $W$ is defined in Sec. \ref{Tree Structure Sparsity in Convex Optimization}.

Since the size of $x$ is $n=N^2$, the storage cost for the matrices $F$, $R$, $\mathcal{W}$, and $W$ is $n^2=N^4$, and is $m\times n$ for matrix $D$.
For example, if the image we aim to reconstruct is of size $128\times128$, then the matrices, $F$, $R$, $\mathcal{W}$, $W$, cost around $8$GB memory in total.
However, if we consider that the permutation matrix $R$ and diagonal matrix $W$ can be represented by $n=N^2$ entries,
around $4$GB are enough.
As a result, the limited storage leads to the restriction of image size.

In order to speed up the computation of linear transformation ({\em i.e.}, matrix-vector multiplications here), we resort to linear operator in MATLAB or reformulation, as described below.
In other words, the memories required to store the matrices, mentioned above, can be remarkably reduced accordingly such that our method can be adaptive to large images.
\begin{itemize}
\item Weighted matrix $W\in\mathbb{R}^{n\times n}$:\\
$Wx=w\circ x$, where $w=vec(W)$ and $\circ$ denotes the
Hadamard product. The memory cost of $W$ is defined as $cost_M(W)=n$ due to Hadamard operation.

\item Inverse wavelet matrix $\mathcal{W}\in\mathbb{R}^{n\times n}$:\\
$\mathcal{W}x=\left(\mathbb{W}_2\otimes\mathbb{W}_1\right)x
=vec(\mathbb{W}_1X\mathbb{W}_2^T)$.

We can see that $cost_M(\mathcal{W})=2n$ due to the use of Kronecker product.

\item Random permutation matrix $R\in\mathbb{R}^{n\times n}$:\\
$Rx=\mathcal{R}(x)$, where the operator $\mathcal{R}(\cdot)$ randomly permutes the indices of vector $x$.
Thus, $cost_M(\mathcal{R})=n$.

\item Discrete Cosine Transform (DCT) $F\in\mathbb{R}^{n\times n}$:\\
$F(x)=dct(x)$, where DCT can be calculated by the dct operator in MATLAB, which is speeded up by Fast Fourier Transform.
Thus, we have $cost_M(F)=n$.

\item Partial random permutation $D\in\mathbb{R}^{m\times n}$:\\
$Dx=\mathcal{D}(x)$, where the operator $\mathcal{D}(\cdot)$ randomly chooses $m$ indices from $n$ components in vector $x$. So $cost_M(\mathcal{D})=m$.
\end{itemize}
Therefore, the memory cost of our method is in total $m+6n$.

On the other hand, the computational complexity for calculating $\hat{A}\hat{x}$ by matrix-vector multiplications and by linear operator are compared as follows.
\begin{itemize}
\item Matrix-vector multiplication: Since\\ $\hat{A}\hat{x}=DFR\mathcal{W}W^{-1}\hat{x}
=\Phi\mathcal{W}W^{-1}\hat{x}$, the time complexity for individual matrix multiplication is:
\begin{center}
\hspace*{-15pt}\begin{tabular}{|c|c|c|c|}
\hline
matrix&$\Phi=DFR$&$\mathcal{W}$&$W^{-1}$\\
\hline
complexity&$O(mn)$&$O(n^2)$&$O(n^2)$\\
\hline
\end{tabular}
\end{center}

\item Linear operations or reformulations: The time complexity for individual operation is:
\begin{center}
\hspace*{-5pt}\begin{tabular}{|c|@{\hspace*{+3pt}}c@{\hspace*{+3pt}}|@{\hspace*{+3pt}}c@{\hspace*{+3pt}}|@{\hspace*{+3pt}}c@{\hspace*{+3pt}}|@{\hspace*{+3pt}}c@{\hspace*{+3pt}}|@{\hspace*{+3pt}}c@{\hspace*{+3pt}}|}
\hline
operation&$D$&$F$&$R$&$\mathcal{W}$&$W^{-1}$\\
\hline
complexity
&{\footnotesize$O(m)$}
&{\footnotesize$O(n\log(n))$}
&{\footnotesize$O(n)$}
&{\footnotesize$O(n^{3/2})$}
&{\footnotesize$O(n)$}\\
\hline
\end{tabular}
\end{center}
\end{itemize}

We can see that the computational complexity $O(n^2)$ of matrix-vector multiplications is reduced to $O(n^{3/2})$ of linear operations.



\subsection{Convergence Analysis}
In this section, we study the convergence of the solution sequence $\{x_n\}$, which is generated by Algorithm \ref{alg1}.
Based on \cite{W.Yin08}, we choose $$\tau \in \left(0,2/\hat{\lambda}_{\max}\right),$$ which guarantees that both two functions in Eq. (\ref{fixed point}) are nonexpansive, where $$\hat{\lambda}_{\max}:= \max_{x}\lambda_{\max}{\nabla}^2\left(\frac{1}{2}\|y-Ax\|^2_2\right).$$

\begin{Thm}\label{thm1}
Let $\{x_n\}$ be a sequence generated by
Algorithm \ref{alg1}. Assume that $\mathcal{J}\neq\emptyset$.
Then $\{x_n\}$ converges to a point in $\mathcal{J}$.
\end{Thm}

The strategy of proving Algorithm \ref{alg1} is to prove that the sequence $\{x_n\}$ is Cauchy, with the fact that every Cauchy sequence converges in complete space $\mathbb{R}^n$.

By Theorem \ref{thm1}, the sequence $\left\{x_n\right\}$ generated by Algorithm \ref{alg1} converges to a fixed point $x$ of
the fixed point equation Eq. (\ref{J}), {\em i.e.}, the optimal solution of Eq. (\ref{Lasso}).

Moreover, Algorithm \ref{alg2} is designed based on Algorithm \ref{alg1} with the weighted matrix constructed in terms of tree structure, that is, Algorithm \ref{alg2} is conducted by iteratively performing Algorithm \ref{alg1} $S-1$ times.
Thus, the convergence of Algorithm \ref{alg2} is guaranteed by the convergence of Algorithm \ref{alg1}.
We show the normalized function errors $\frac{\mathcal{F}(x^k)-\mathcal{F}(x^*)}{\mathcal{F}(x^k)}$ vs. number of iterations in Fig. \ref{error curve}, where the measurement rates range from $10\%$ to $30\%$.
We can see that the algorithm converges more and more fast (with less number of iterations) as the measurement rates increase.

\begin{figure*}[!hbpt]
\hspace*{+5pt}
\begin{minipage}[b]{.24\linewidth}
\includegraphics[width=4cm]{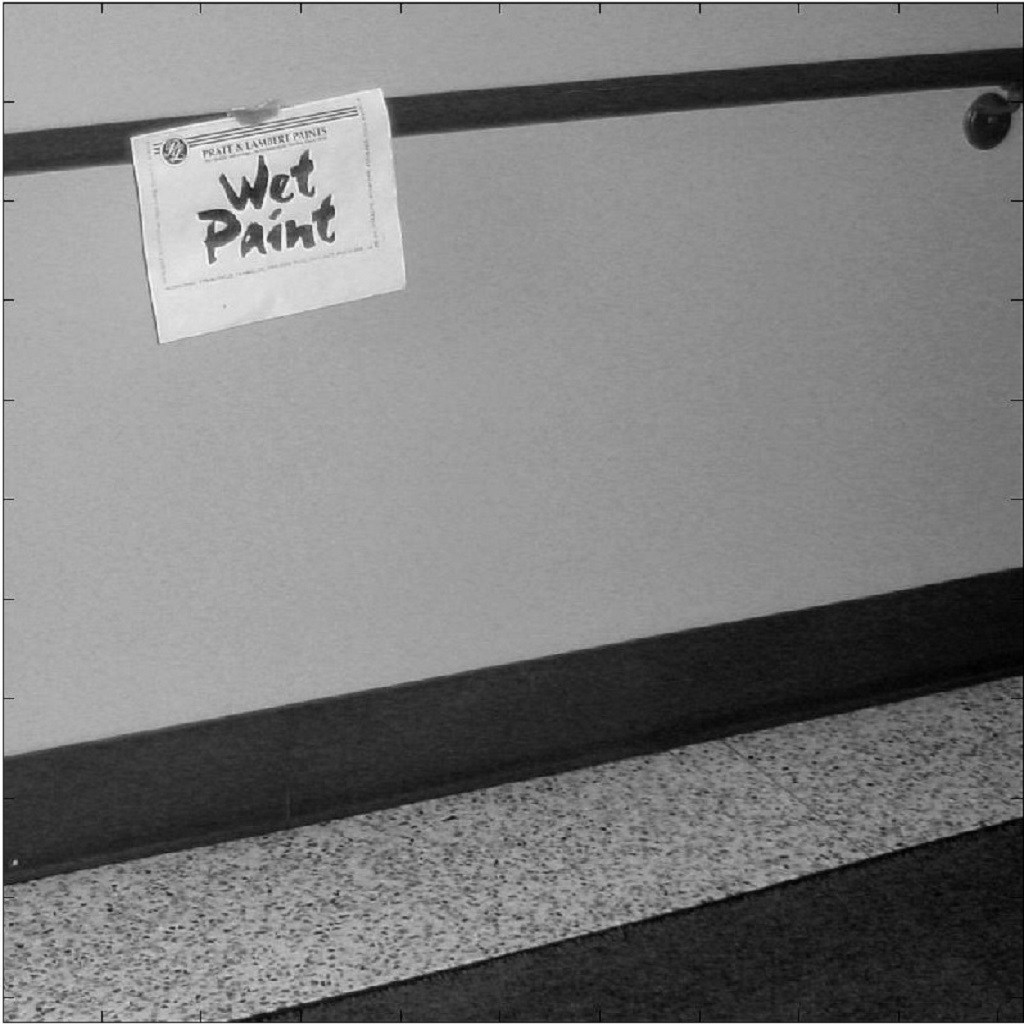}
  \hspace*{-5pt}\centerline{(a1)}
\end{minipage}
\begin{minipage}[b]{.24\linewidth}
\includegraphics[width=4cm]{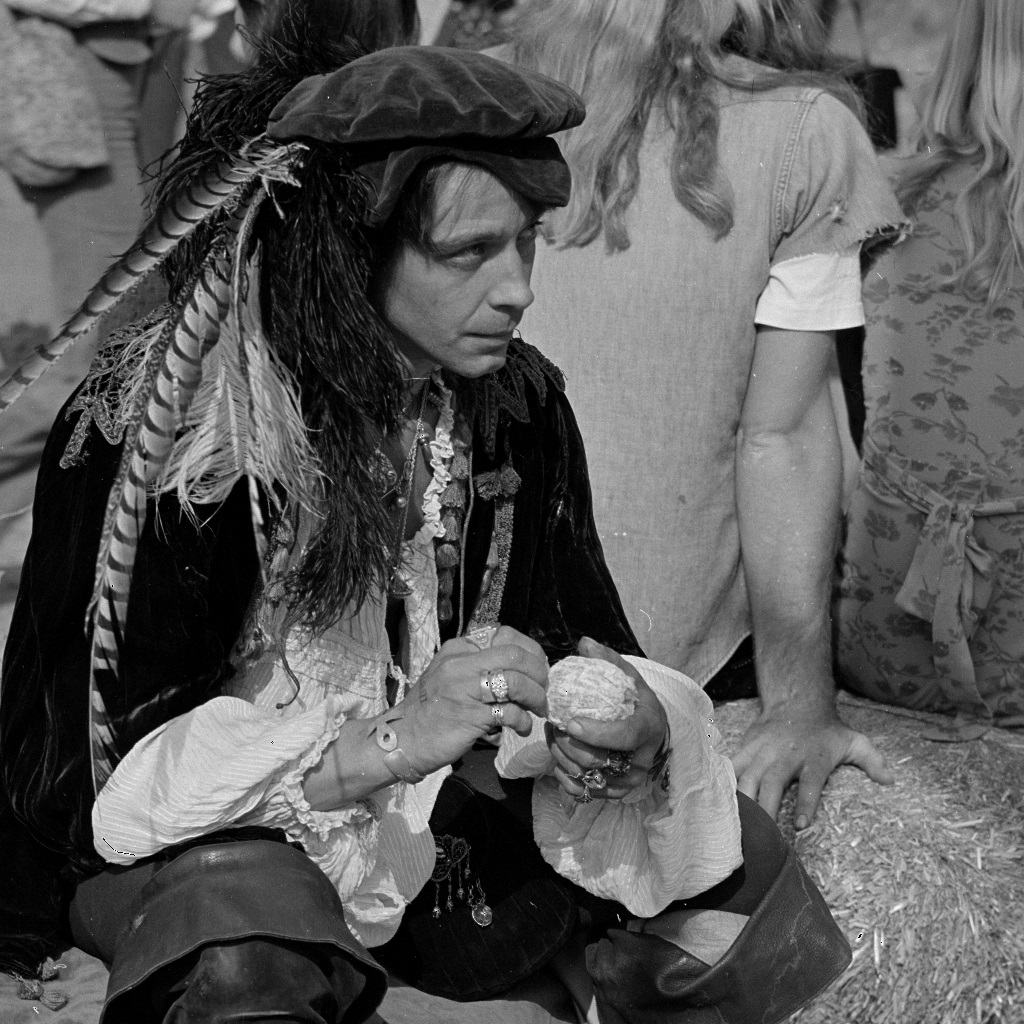}
  \hspace*{-5pt}\centerline{(a2)}
\end{minipage}
\begin{minipage}[b]{.24\linewidth}
\includegraphics[width=4cm]{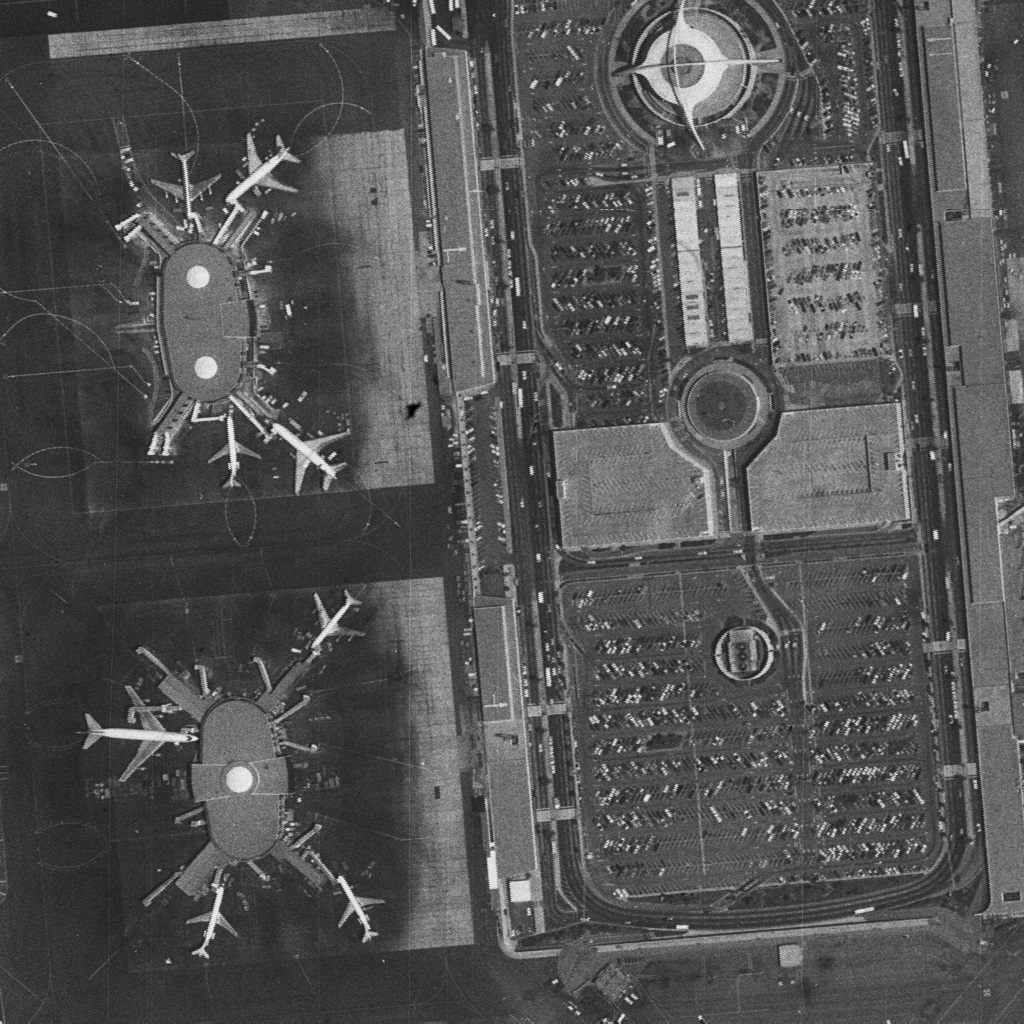}
  \hspace*{-5pt}\centerline{(a3)}
\end{minipage}
\begin{minipage}[b]{.24\linewidth}
\includegraphics[width=4cm]{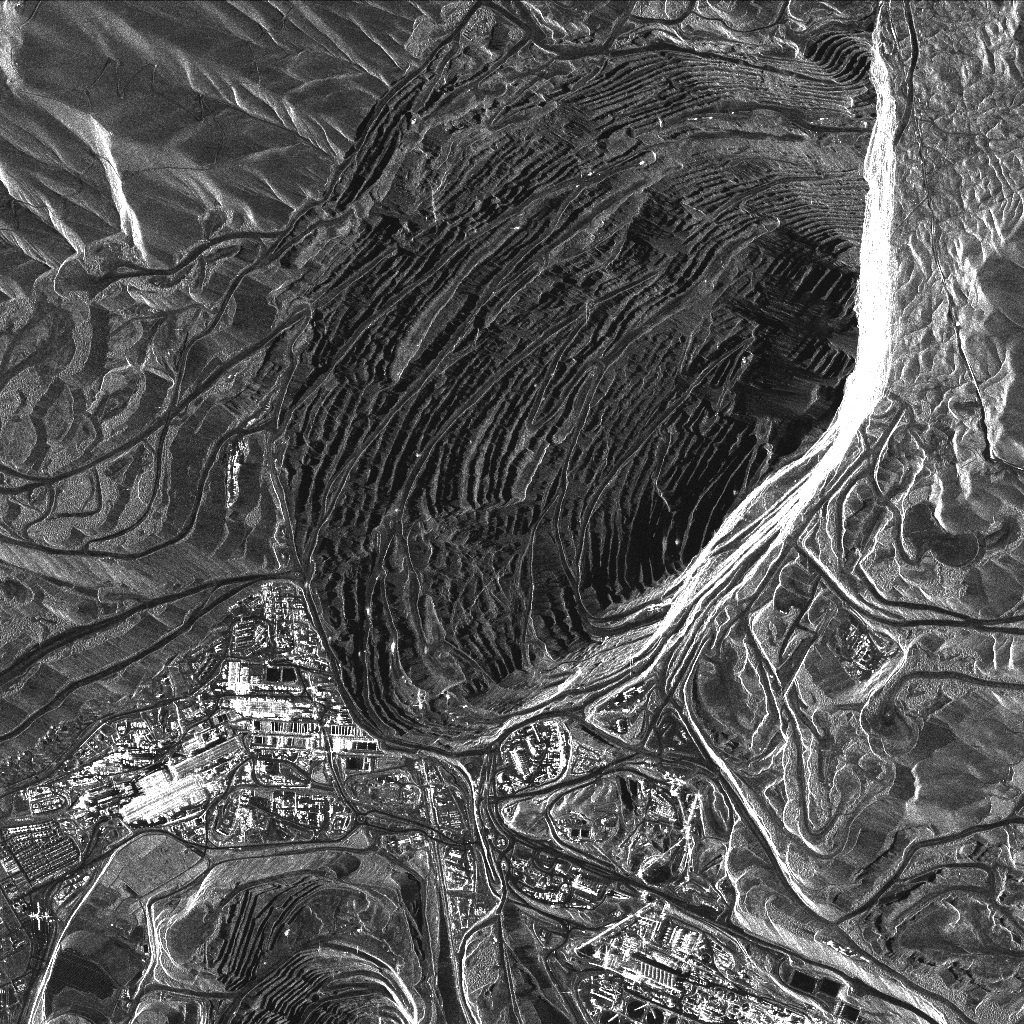}
  \hspace*{-5pt}\centerline{(a4)}
\end{minipage}\\
\begin{minipage}[b]{.24\linewidth}
\includegraphics[width=4.8cm]{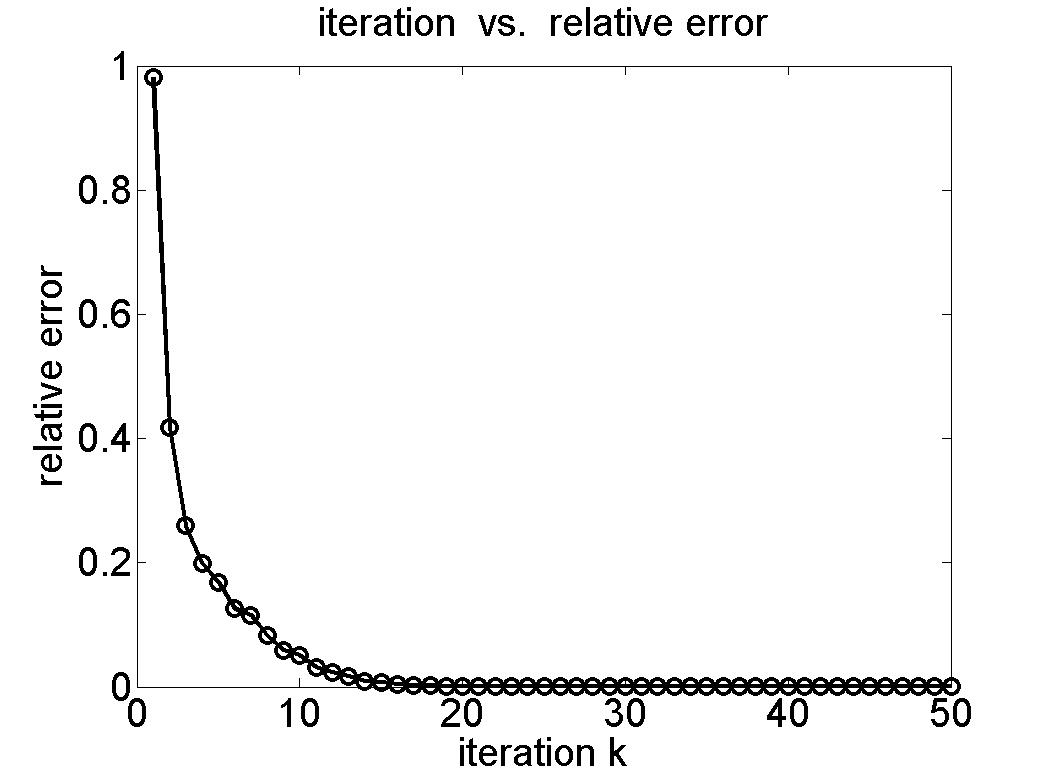}
  \centerline{(b1)}
\end{minipage}
\begin{minipage}[b]{.24\linewidth}
\includegraphics[width=4.8cm]{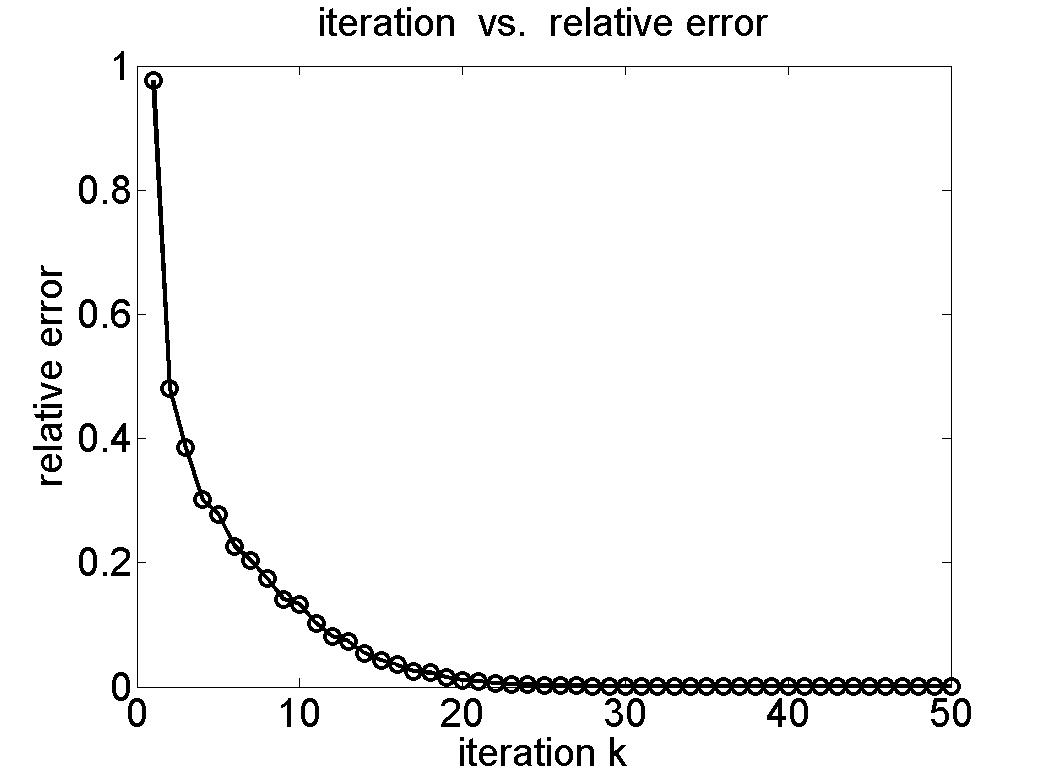}
  \centerline{(b2)}
\end{minipage}
\begin{minipage}[b]{.24\linewidth}
\includegraphics[width=4.8cm]{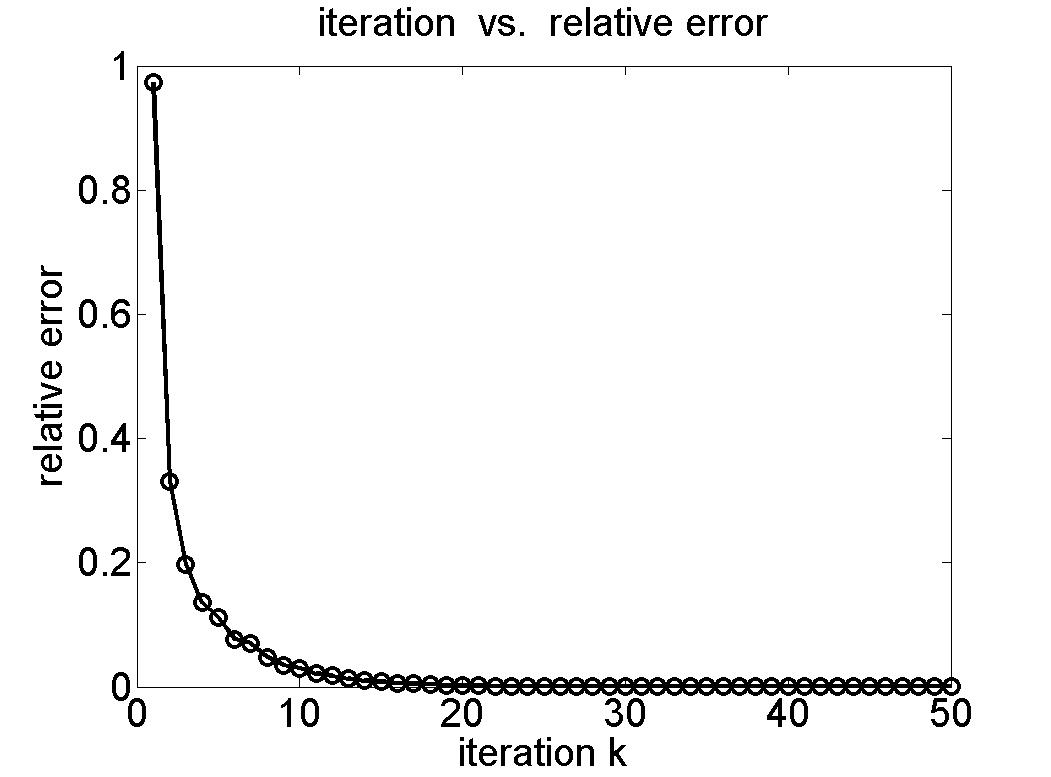}
  \centerline{(b3)}
\end{minipage}
\begin{minipage}[b]{.24\linewidth}
\includegraphics[width=4.8cm]{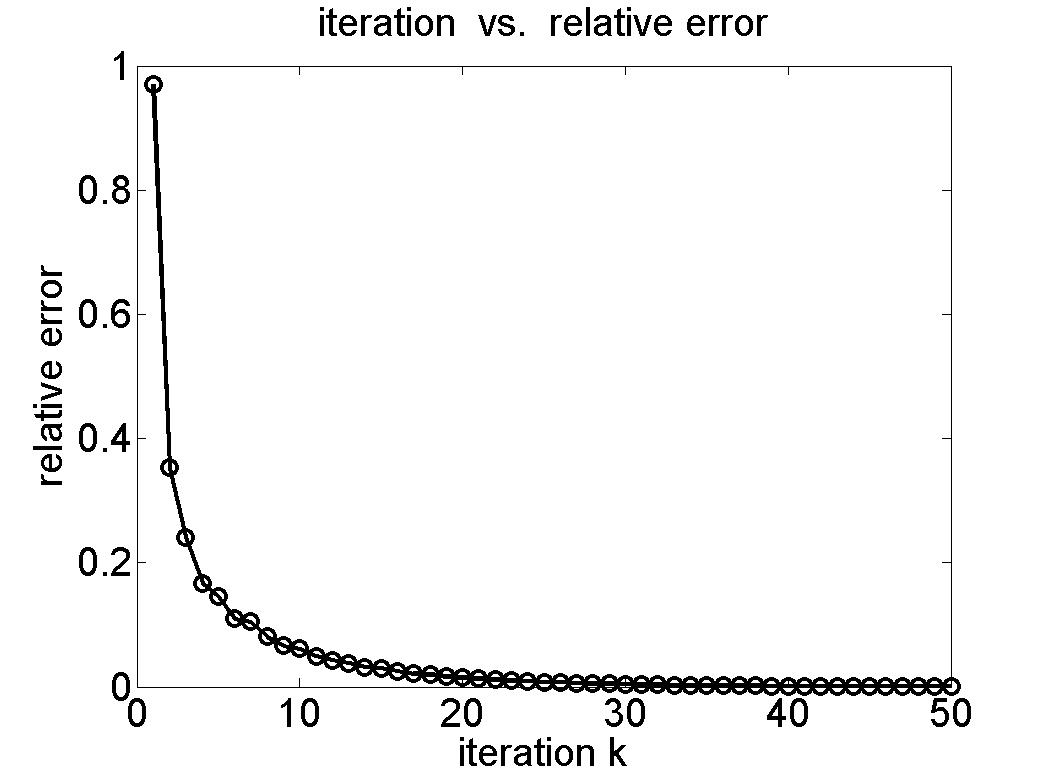}
  \centerline{(b4)}
\end{minipage}\\
\begin{minipage}[b]{.24\linewidth}
\includegraphics[width=4.8cm]{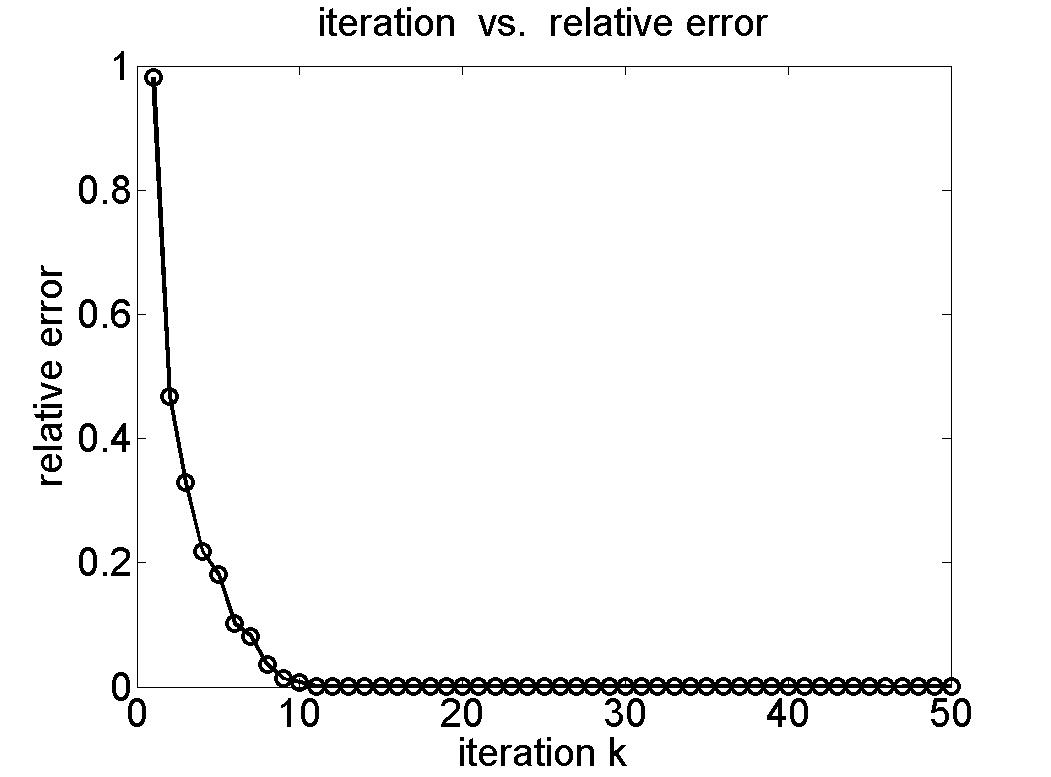}
  \centerline{(c1)}
\end{minipage}
\begin{minipage}[b]{.24\linewidth}
\includegraphics[width=4.8cm]{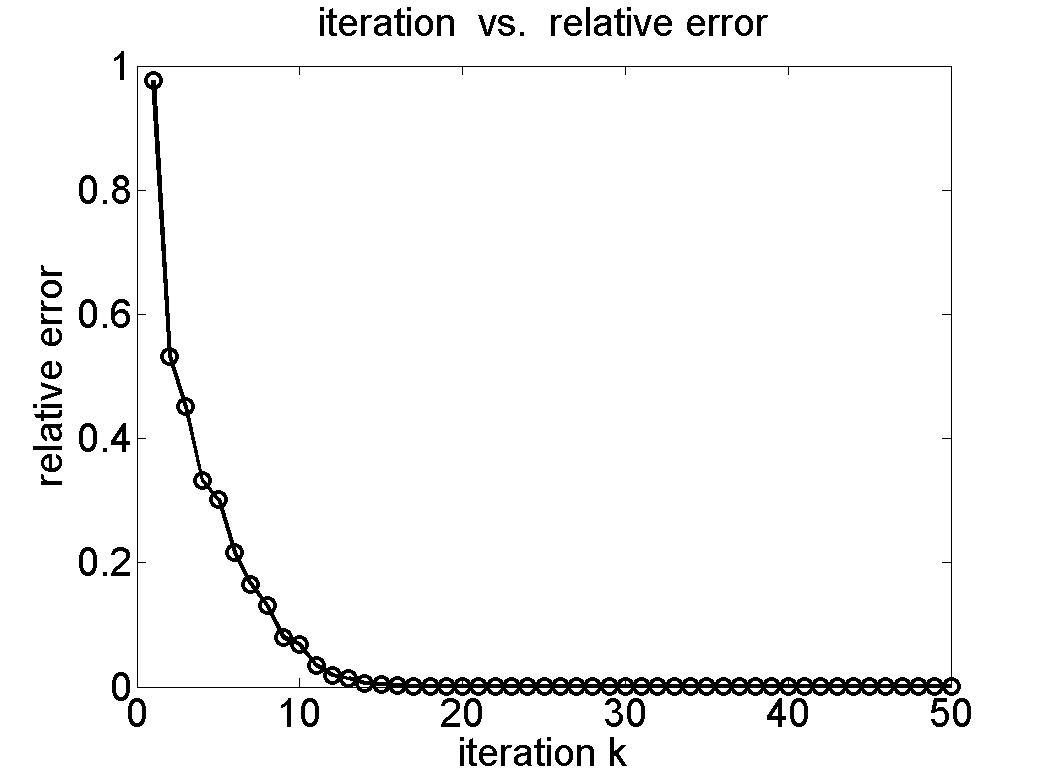}
  \centerline{(c2)}
\end{minipage}
\begin{minipage}[b]{.24\linewidth}
\includegraphics[width=4.8cm]{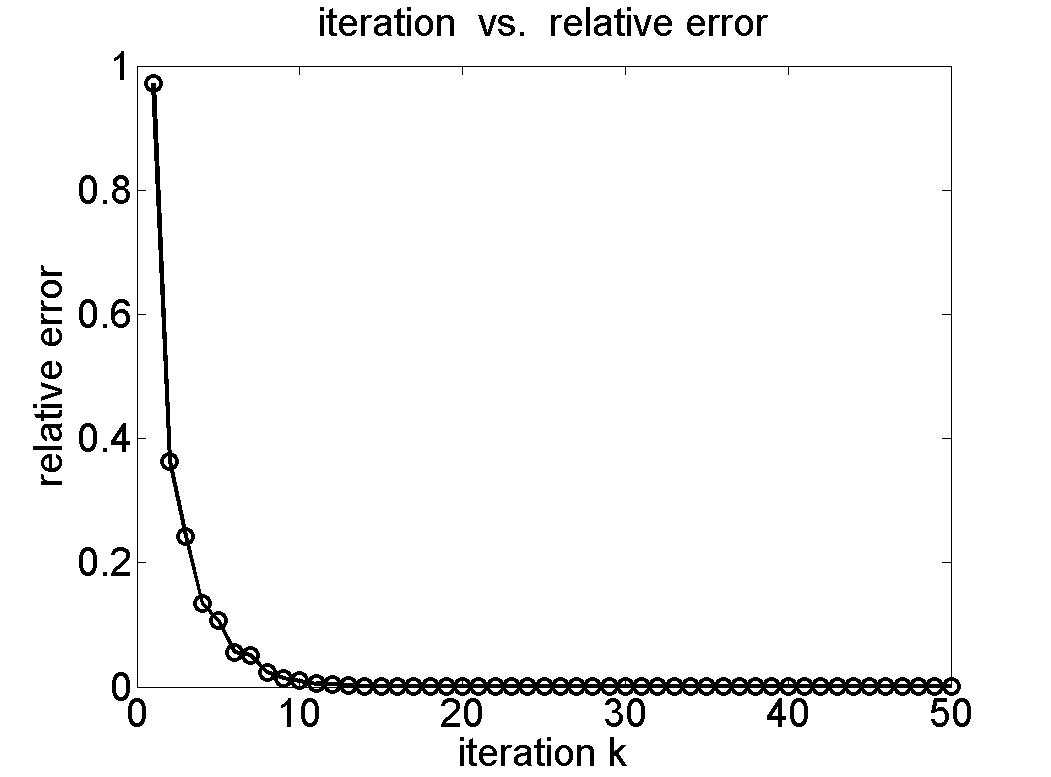}
  \centerline{(c3)}
\end{minipage}
\begin{minipage}[b]{.24\linewidth}
\includegraphics[width=4.8cm]{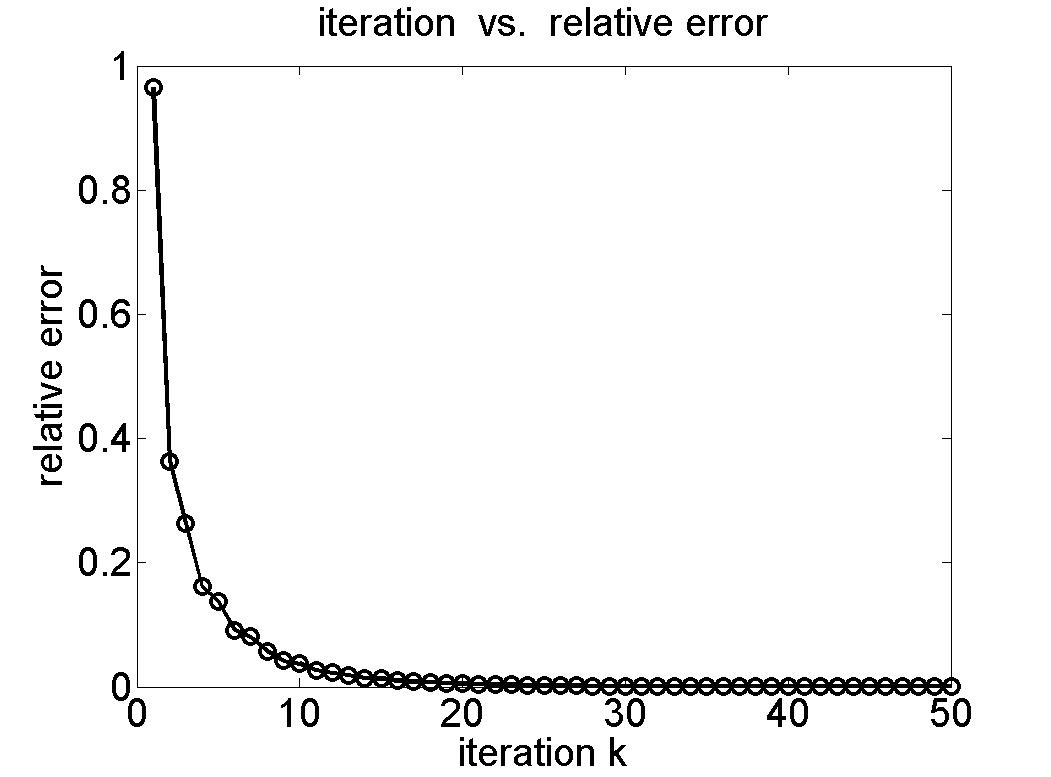}
  \centerline{(c4)}
\end{minipage}\\
\begin{minipage}[b]{.24\linewidth}
\includegraphics[width=4.8cm]{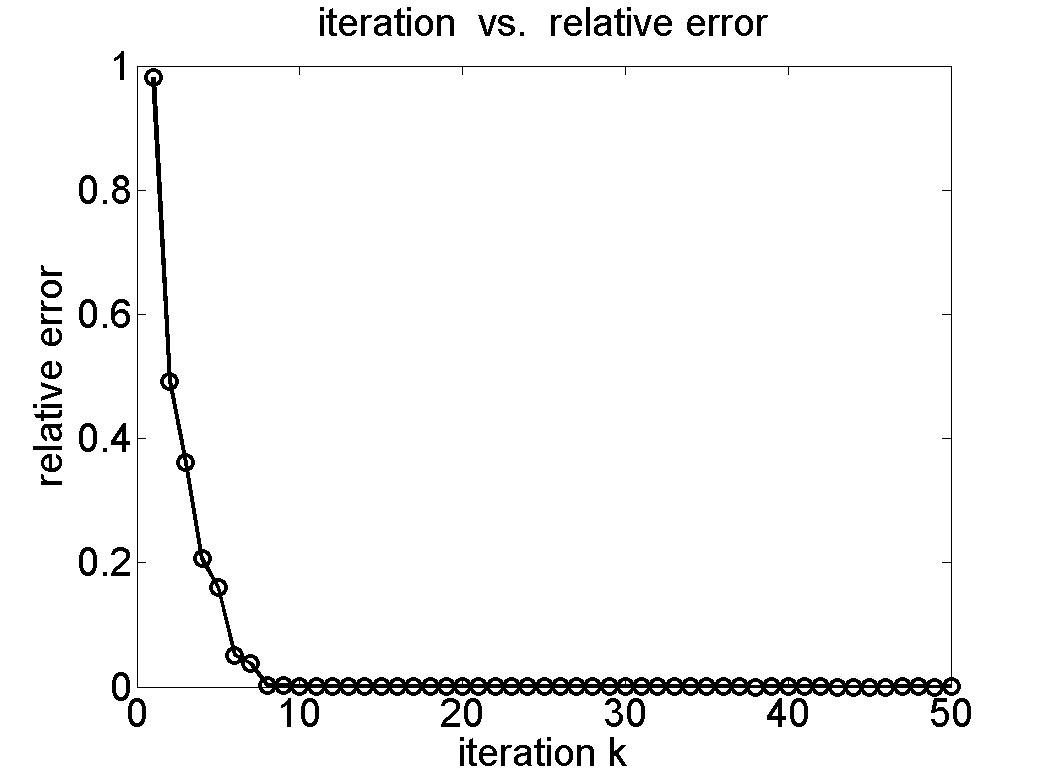}
  \centerline{(d1)}
\end{minipage}
\begin{minipage}[b]{.24\linewidth}
\includegraphics[width=4.8cm]{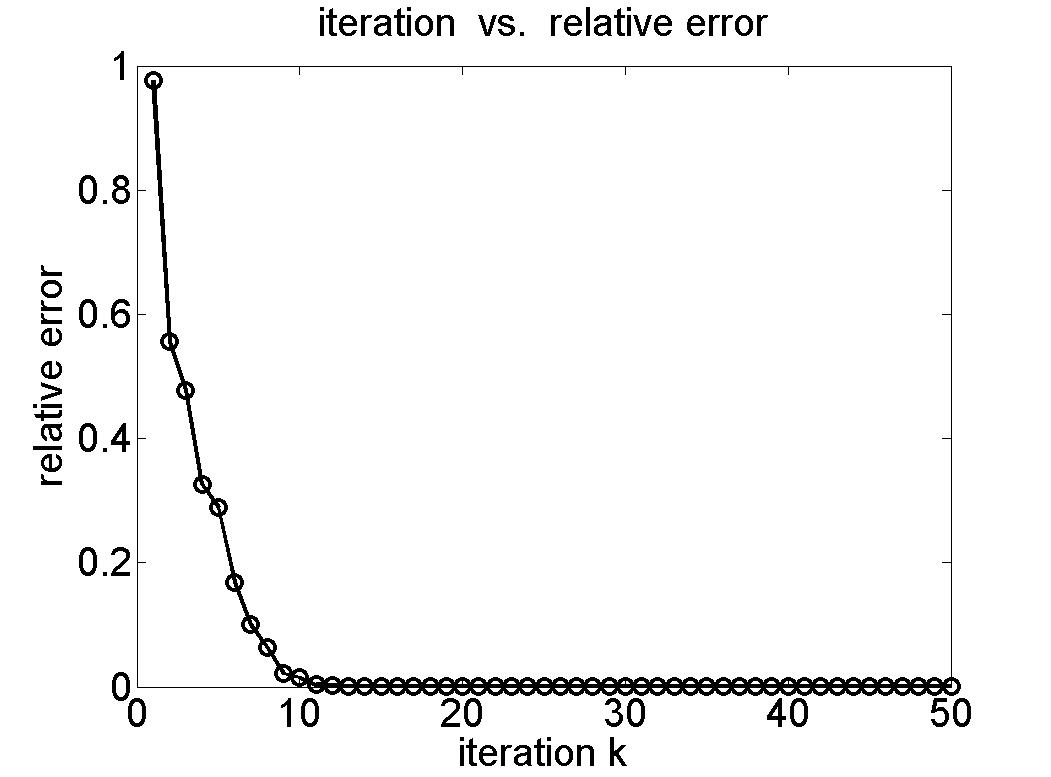}
  \centerline{(d2)}
\end{minipage}
\begin{minipage}[b]{.24\linewidth}
\includegraphics[width=4.8cm]{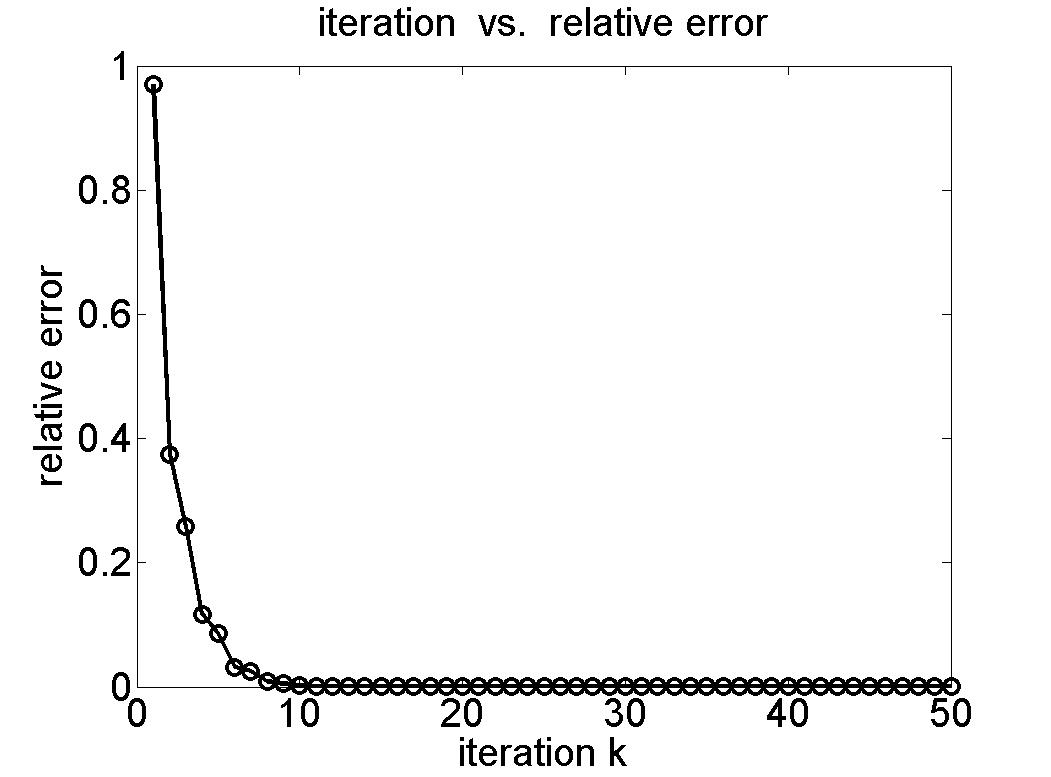}
  \centerline{(d3)}
\end{minipage}
\begin{minipage}[b]{.24\linewidth}
\includegraphics[width=4.8cm]{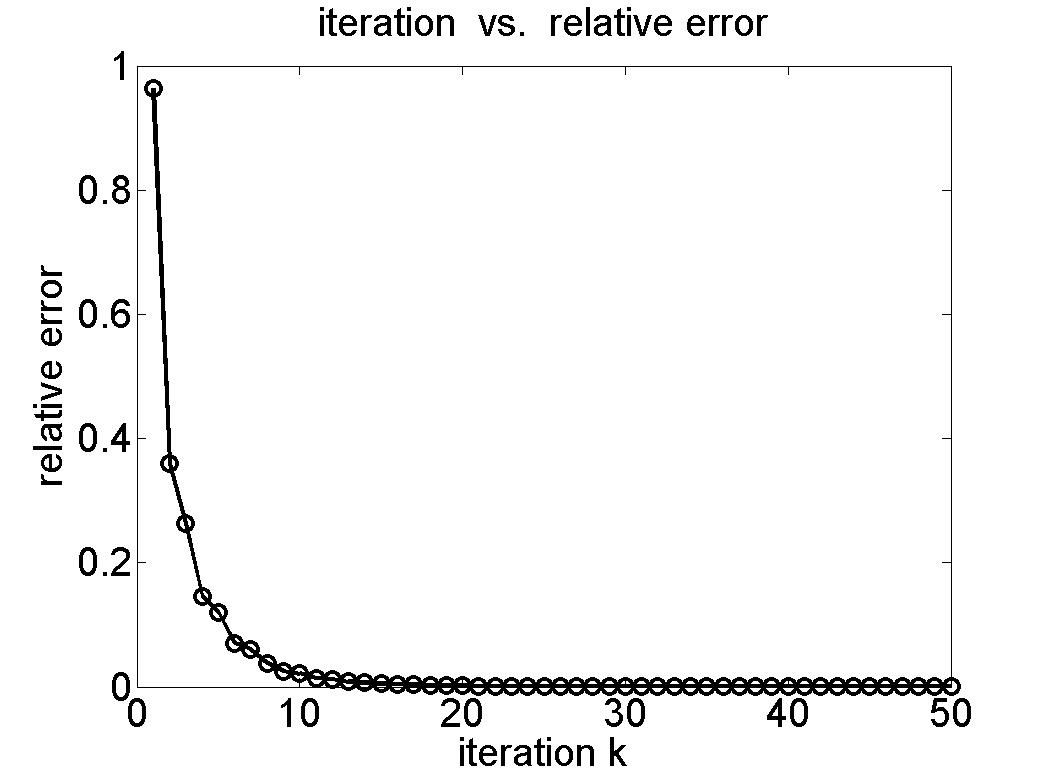}
  \centerline{(d4)}
\end{minipage}
\caption{The comparison of normalized function error
$\frac{\mathcal{F}(x^k)-\mathcal{F}(x^*)}{\mathcal{F}(x^k)}$ ($y$-axis) versus number of iteration ($x$-axis) under different measurement rates in row (b) $10\%$; row (c) $20\%$; and row (d) $30\%$ for the corresponding four images in row (a).}
\label{error curve}
\end{figure*}


\section{Simulation Results}
All simulations were conducted on PC equipped with Windows 7 with 3.4 GHz Intel Core i7 CPU and 8GB RAM.
We compare the proposed method, N-BOMP \cite{Caiafa13}, and \cite{Caiafa14} in terms of CPU time for CS recovery and reconstruction quality.
Here, we take 2D image of large sizes (ranging from $1024\times 1024$ to $4096\times 4096$) for compressive sensing and recovery, and show some results for images with different degree of sparsity in this section.
The images are shown in Fig. \ref{error curve} (a).
They are (from left to right) Paint, Man (the first two adopted in \cite{Caiafa13}), airport (a commonly used image in the literature), and a synthetic aperture radar (SAR) image (a TerraSAR-X image downloaded from http://www.geo-airbusds.com/en/23-sample-imagery).

\subsection{Parameter Setting}
The sensing matrix was a Structurally Random Matrix \cite{Do12}, the Daubechies wavelet was used as the sparsifying basis, and fixed point method with quasi-Armijo rule was adopted for CS recovery.
Note that the images to be sensed were not pre-processed in advance.


\subsection{Performance Comparison}
As mentioned previously, the four images of size $1024\times 1024$ shown in Fig. \ref{error curve} (a) were used as the targets for sensing.
Several measurement rates (MRs), ranging from $5\%$ to $50\%$, were selected for simulations.
The CPU time, the reconstruction quality measured in Peak-Signal-to-Noise-Ratio (PSNR), and the reconstruction quality measured in structural similarity (SSIM) \cite{ssim} were adopted as the criteria for comparison among our method (Algorithm \ref{alg2}), N-BOMP \cite{Caiafa13}, and \cite{Caiafa14}.
Some results are shown in the following figures.

\begin{figure}[!h]
\hspace*{-0.3cm}\includegraphics[width=9.5cm]{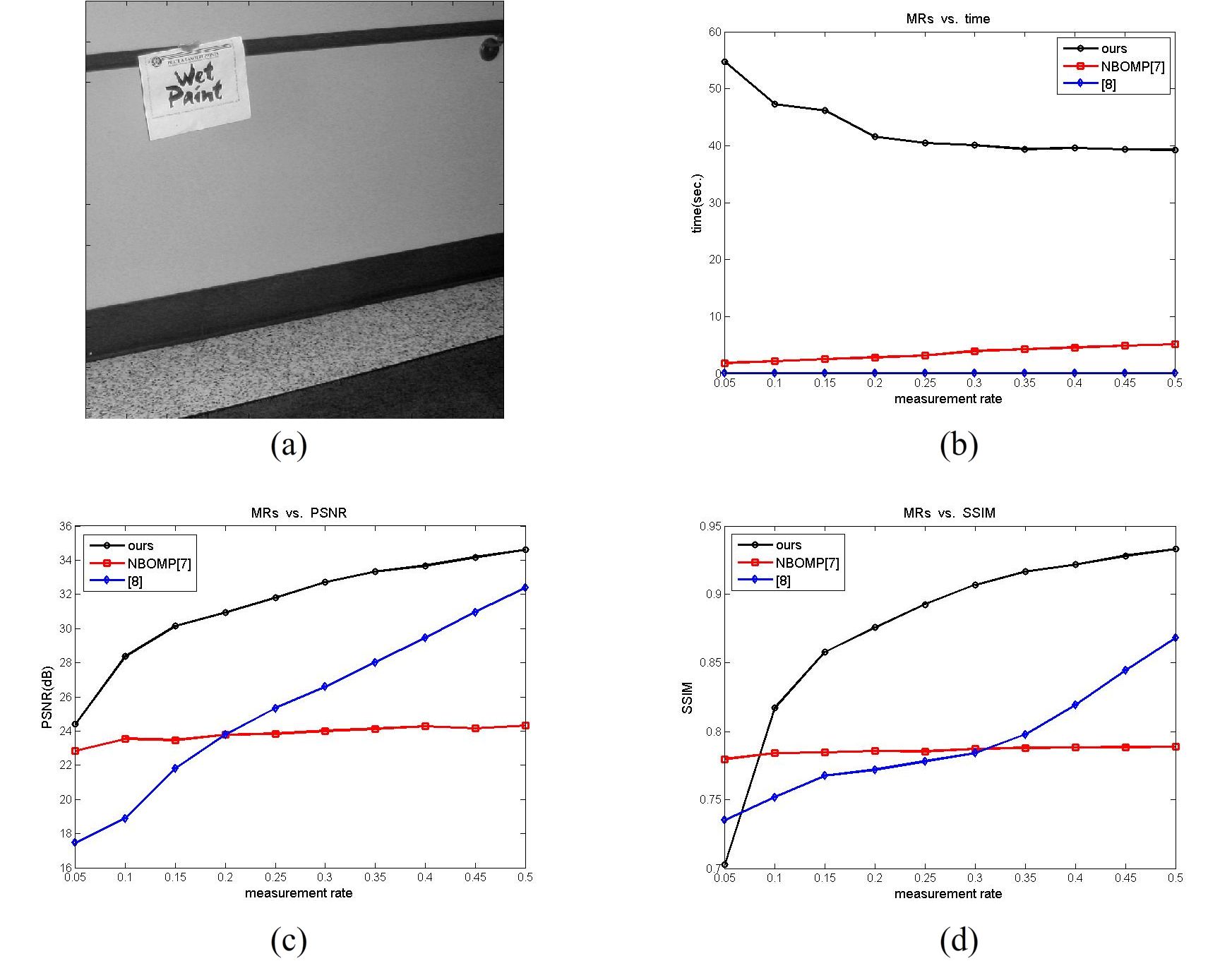}
\caption{Comparison of time/PSNR/SSIM of the ``Paint'' image among Algorithm \ref{alg2}, NBOMP \cite{Caiafa13}, and \cite{Caiafa14}.}
\label{paint}
\end{figure}

\begin{figure}[!h]
\hspace*{-0.3cm}\includegraphics[width=9.5cm]{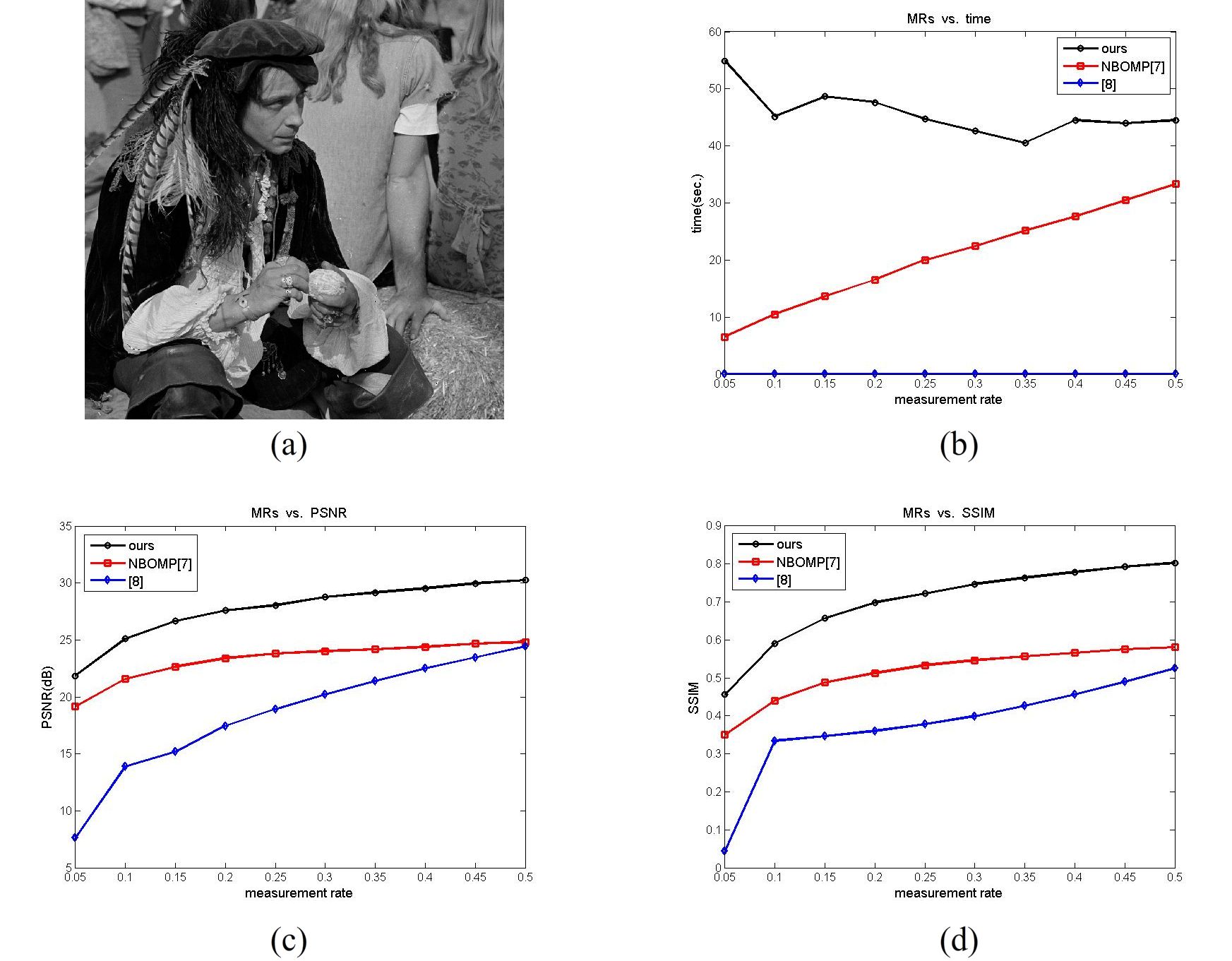}
\caption{Comparison of time/PSNR/SSIM of the ``Man'' image among Algorithm \ref{alg2}, NBOMP \cite{Caiafa13}, and \cite{Caiafa14}.}
\label{man}
\end{figure}

\begin{figure}[!h]
\hspace*{-0.3cm}\includegraphics[width=9.5cm]{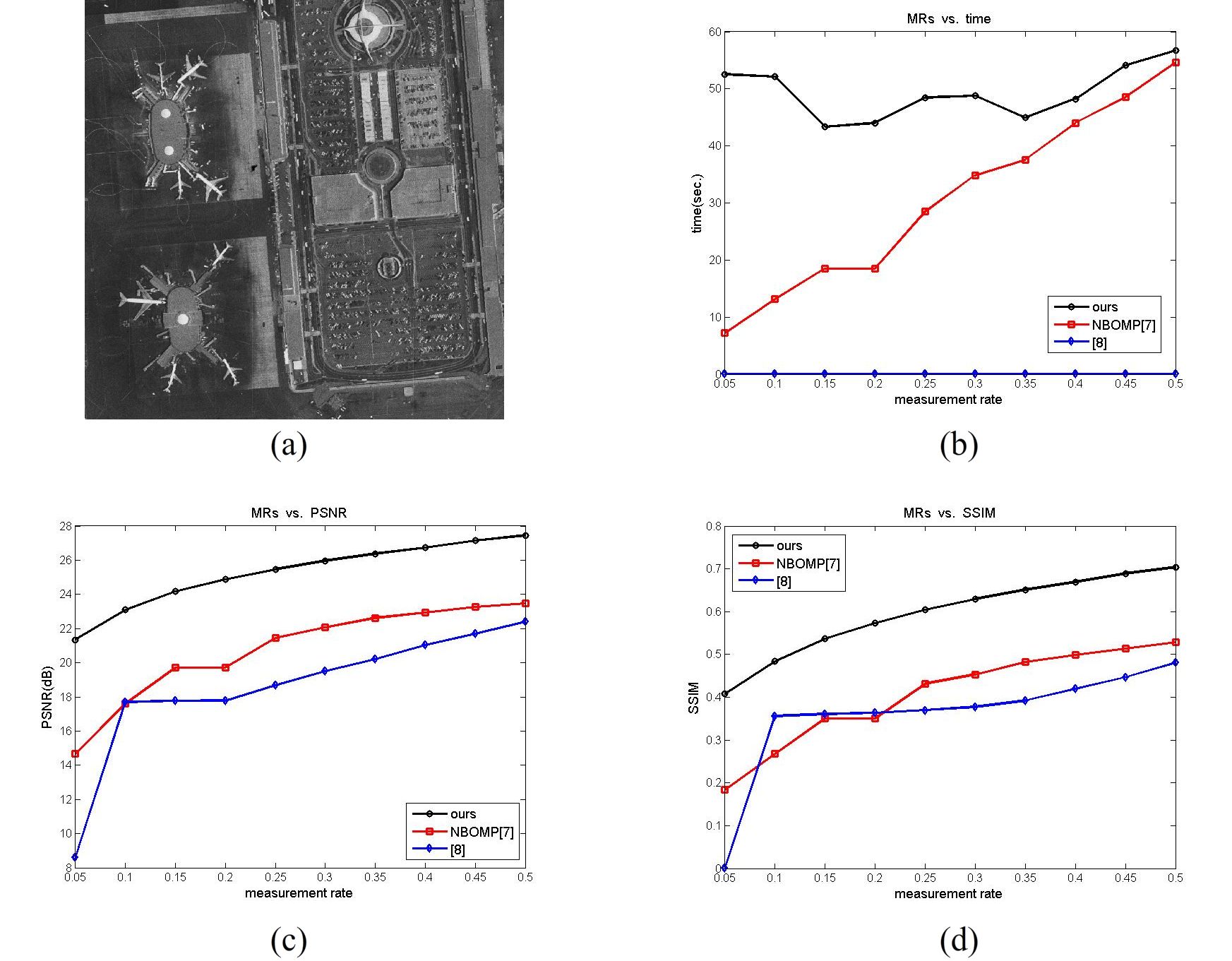}
\caption{Comparison of time/PSNR/SSIM of the ``Airport'' image among Algorithm \ref{alg2}, NBOMP \cite{Caiafa13}, and \cite{Caiafa14}.}
\label{airport}
\end{figure}

\begin{figure}[!h]
\hspace*{-0.3cm}\includegraphics[width=9.5cm]{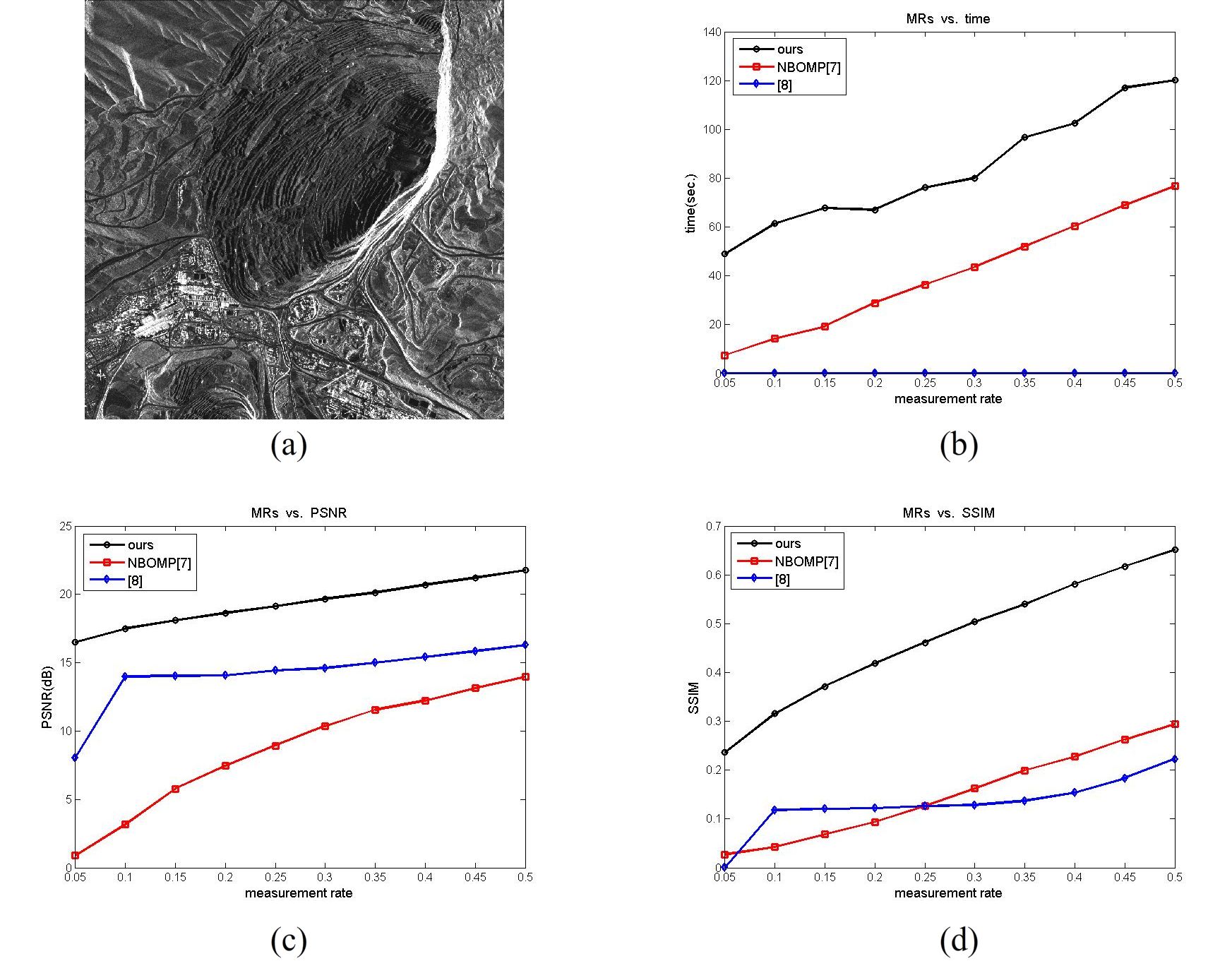}
\caption{Comparison of time/PSNR/SSIM of an SAR image among Algorithm \ref{alg2}, NBOMP \cite{Caiafa13}, and \cite{Caiafa14}.}
\label{SAR}
\end{figure}

From Fig. \ref{paint} to Fig. \ref{SAR}, we can find that under the same measurement rates, our method obtains the best reconstruction qualities in term of PSNR and SSIM within reasonable recovery time.
Although our method needs more computational time than the other two methods, it does not rely on any impractical assumptions and restrictions (block sparsity or low rank) and is flexible in practice.
In particular, we do not think it is meaningful and impressive to quickly yield inaccurate results.


Furthermore, for visual comparison, we actually observe that the image details can be well recovered by our method.
This also explains the usefulness of tree-structure sparsity pattern imposed in the $l_1$-norm minimization formulation (Eq. (\ref{Lasso_weighted})).
In fact, our simulations also show that when weighting is imposed, the PSNR gain of $0.5\sim 1$ dB can be obtained, when compared to its non-weighting counterpart.

Finally, when the popularly used CVX software is considered for CS recovery, our results show that basically CVX cannot deal with images larger than $100\times 100$ due to high cost of memory and computation.





\section{Conclusion}
Compressive sensing of large-scale images is remarkably challenging due to the constraints of storage and computation.
In this paper, we propose a new method of large-scale image compressive sensing based on exploring a fixed-point weighted-LASSO algorithm without depending on any assumption or preprocessing of sparsity pattern in images.
Convergence analysis is also provided to confirm the convergence of our iterative scheme.


%

\appendices
\section{Proof of Theorem 1}
To prove Theorem 1, we need the following Lemmata.
\begin{Lem}
\label{nex}
The operator $S_{\tau\mu}(\cdot)$
in Eq. (\ref{fixed point}) is nonexpansive
(Lemma 3.2 in \cite{W.Yin08});
$G_{\tau}(\cdot)$
is nonexpansive for an appropriate parameter $\tau$
(Lemma 4.1 and Eq. (4.3) in \cite{W.Yin08});
and thus $\left(S_{\tau\mu}\circ G_{\tau}\right)(\cdot)$ is as well. Moreover, $\left(S_{\tau\mu}\circ G_{\tau}\right)(\cdot)$ is continuous. 
\end{Lem}
\begin{Lem}
\label{para}
In a normed vector space $(X,\|\cdot\|)$, the following identical equation is called the parallelogram law:
$$2\|x\|^2+2\|y\|^2=\|x+y\|^2+\|x-y\|^2, \quad \forall x,y\in X.$$
\end{Lem}

\begin{Lem}
\label{liminf}
Let $\{a^k\},\{b^k\}$ be non-negative sequences,
and $\{a^k\}$ does not converge to 0.
If $\displaystyle\lim_{k\rightarrow \infty}a^kb^k=0$, then $\displaystyle\liminf_{k\rightarrow \infty} b^k=0.$
\end{Lem}
\begin{proof}
Since $\{a^k\}$ does not converge to 0, there exists $\epsilon_1>0$ {\em s.t.} for all $j\in \mathbb{N}$,
we can find $n_j \geq j$ so that $|a^{n_j}-0|>\epsilon_1$.\\
We have
$$0=\lim_{j\rightarrow \infty}a^{n_j}b^{n_j}\geq \lim_{j\rightarrow \infty}\epsilon_1b^{n_j}\geq 0,$$
which implies
$\lim\limits_{j\rightarrow \infty}b^{n_j}=0$.
In other words, the limit of subsequence $\left\{b^{n_j}\right\}$
exists, and hence
$$\lim_{k\rightarrow\infty}\inf_{n_j\geq k}b^{n_j}
=\lim_{j\rightarrow\infty}b^{n_j}.$$

\noindent
By the fact that $\{b^j\}_{j\geq k}\supset \{b^{n_j}\}_{n_j\geq k}$,
we have
$$\inf_{j\geq k}b^j \leq \inf_{n_j \geq k}b^{n_j}.$$
Thus
$$0
\leq\liminf_{k\rightarrow\infty}b^k
=\lim_{k\rightarrow\infty}\inf_{j\geq k}b^j
\leq\lim_{k\rightarrow\infty}\inf_{n_j\geq k}b^{n_j}
=\lim_{j\rightarrow\infty}b^{n_j}
=0.
$$
Finally, we have $\displaystyle\liminf_{k\rightarrow\infty}b^k=0$.
\end{proof}

\noindent
In the following, we prove Theorem \ref{thm1}.

\begin{proof}
First we show that
the limit $\displaystyle\lim_{k\rightarrow \infty}\|x^k-z\|$ exists.
Let $z \in \mathcal{J}$, and
define $\mathcal{P}^k=S_{\tau\mu}\circ G_{\tau}(x^k)$. Then
\begin{eqnarray*}
\begin{array}{r@{\hspace*{+4pt}}l}
&\left\|x^{k+1} - z\right\|_2\\
=&\left\|x^k +\sigma_k\left(\mathcal{P}^k-x^k\right) - z\right\|_2\\
=&\left\|(1-\sigma_k)x^k+\sigma_k\mathcal{P}^k-z\right\|_2\\
=&\left\|(1-\sigma_k)(x^k-z)+\sigma_k\left(\mathcal{P}^k-z\right)\right\|_2\\
\leq&(1-\sigma_k)\left\|x_k-z\right\|_2
+\sigma_k\left\|\mathcal{P}^k-z\right\|_2\\
\leq &(1-\sigma^k)\left\|x^k-z\right\|_2+\sigma_k\left\|x^k-z\right\|_2\\
& \mbox{(since $z\in\mathcal{J}$ and by Lemma \ref{nex})}\\
=&\left\|x^k-z\right\|_2.
\end{array}
\end{eqnarray*}
Hence, the sequence $\{\|x^k-z\|\}$ is monotone decreasing and $\displaystyle\lim_{k\rightarrow \infty}\|x^k-z\|$ exists.

Second we show that the sequence $\left\{x^k\right\}$ is Cauchy.
By Lemma \ref{para}, we have
\begin{eqnarray}
\label{eq1}
\begin{array}{r@{\hspace*{+4pt}}l}
&\|x^h-x^k\|^2\\
=&\|(x^h-z)-(x^k-z)\|^2\\
=&2\|x^h-z\|^2+2\|x^k-z\|^2-\|x^h+x^k-2z\|^2.
\end{array}
\end{eqnarray}
Let $\displaystyle\lim_{k\rightarrow \infty}\|x^k-z\|=c$.
We have
\begin{eqnarray*}
\begin{array}{r@{\hspace*{+4pt}}l}
&\left|\|x^h+x^k-2z\|-2c\right|\\
=&|\|(x^h-z)+(x^k-z)\|-c-c|\\
\leq &|\|x^h-z\|-c|+|\|x^k-z\|-c|
\end{array}
\end{eqnarray*}
and
\begin{eqnarray*}
\begin{array}{r@{\hspace*{+4pt}}l}
&\lim\limits_{h,k\rightarrow\infty}\left|\|x^h+x^k-2z\|-2c\right|\\
\leq &\lim\limits_{h\rightarrow\infty}|\|x^h-z\|-c|+
\lim\limits_{k\rightarrow\infty}|\|x^k-z\|-c|\\
=&0,
\end{array}
\end{eqnarray*}
which means
\begin{eqnarray}
\label{eq2}
\lim\limits_{h,k\rightarrow\infty}\|x^h+x^k-2z\|=2c.
\end{eqnarray}

\noindent
By Eq. (\ref{eq1}) and Eq. (\ref{eq2}), we obtain
$$\lim_{h,k \rightarrow \infty}\|x^h-x^k\|_2^2
=2c^2+2c^2-(2c)^2=0,$$
and hence the sequence $\{x^k\}$ is Cauchy. Therefore, $\{x^k\}$ is a convergent sequence.

Third we prove that limit $x^*=\lim\limits_{k\rightarrow\infty}x^k
\in\mathcal{J}$.
Since $x^{k+1}=x^k+\sigma_k(\mathcal{P}^k-x^k)$,
we have
\begin{eqnarray}
\label{eq3}
\sigma_k\|\mathcal{P}^k-x^k\|=\|x^{k+1}-x^k\|,
\end{eqnarray}
and then
\begin{eqnarray*}
\begin{array}{r@{\hspace*{+4pt}}l}
\lim\limits_{k\rightarrow \infty}\sigma_k\left\|\mathcal{P}^k-x^k\right\|
&=\lim\limits_{k \rightarrow \infty}\left\|x^{k+1}-x^k\right\|\\
&=\left\|\lim\limits_{k\rightarrow \infty}(x^{k+1}-x^k)\right\|\\
&=\|x-x\|\\
&=0.
\end{array}
\end{eqnarray*}

\noindent
By the fact that sequence $\left\{\sigma_k\right\}$ is decided
by Step \ref{alg1:quasi Armijo} for each iteration in Algorithm \ref{alg1}, each $\sigma_k$
is mutually independent, and, thus, $\left\{\sigma_k\right\}$
does not converge.
By Lemma \ref{liminf}, we have
\begin{eqnarray}
\label{eq4}
\liminf_{k\rightarrow \infty}\|\mathcal{P}^k-x^k\|=0.
\end{eqnarray}

\noindent
Next we aim to show that
$\displaystyle\lim_{k\rightarrow\infty}\|\mathcal{P}^k-x^k\|$ exists.
Since
\begin{eqnarray*}
\begin{array}{r@{\hspace*{+4pt}}l}
&\left\|\mathcal{P}^{k+1}-x^{k+1}\right\|\\
=&\left\|\mathcal{P}^{k+1}-x^k-\sigma_k(\mathcal{P}^k-x^k)\right\|\\
=&\left\|\mathcal{P}^{k+1}-(1-\sigma_k)x^k-\sigma_k\mathcal{P}^k\right\|\\
=&\left\|\mathcal{P}^{k+1}-\mathcal{P}^k
-(1-\sigma_k)x^k+(1-\sigma_k)\mathcal{P}^k\right\|\\
=&\left\|\mathcal{P}^{k+1}-\mathcal{P}^k+(1-\sigma_k)(\mathcal{P}^k-x^k)\right\|\\
\leq&\left\|\mathcal{P}^{k+1}-\mathcal{P}^k\right\|
+(1-\sigma_k)\left\|\mathcal{P}^k-x^k\right\|,
\end{array}
\end{eqnarray*}
by the fact that $S_{\tau\mu}\circ G_{\tau}$ is nonexpansive,
that is
$$\left\|\mathcal{P}^{k+1}-\mathcal{P}^k\right\|
\leq\left\|x^{k+1}-x^k\right\|.$$
According to Eq. (\ref{eq3}), it follows that
\begin{eqnarray*}
\begin{array}{r@{\hspace*{+4pt}}l}
&\left\|\mathcal{P}^{k+1}-\mathcal{P}^k\right\|
+(1-\sigma_k)\left\|\mathcal{P}^k-x^k\right\|\\
\leq&\|x^{k+1}-x^k\|+(1-\sigma_k)\|\mathcal{P}^k-x^k\|\\
=&\sigma_k\|\mathcal{P}^k-x^k\|+(1-\sigma_k)\|\mathcal{P}^k-x^k\|\\
=&\|\mathcal{P}^k-x^k\|.
\end{array}
\end{eqnarray*}
This implies that $\{\|\mathcal{P}^k-x^k\|\}$ is a decreasing sequence and bounded below by $0$.
Therefore, the sequence $\{\|\mathcal{P}^k-x^k\|\}$ is convergent and Eq. (\ref{eq4}) leads to $$\lim_{k\rightarrow\infty}\|\mathcal{P}^k-x^k\|=0.$$

Finally, according to Lemma \ref{nex}, 
both functions $S_{\tau\mu}\circ G_{\tau}$ and
$\left\|\cdot\right\|$ are continuous.
We have
\begin{eqnarray*}
\begin{array}{r@{\hspace*{+4pt}}l}
0&=\lim\limits_{k\rightarrow \infty}\left\|\mathcal{P}^k-x^k\right\|\\
&=\left\|\lim\limits_{k\rightarrow \infty}(\mathcal{P}^k-x^k)\right\|\\
&=\left\|S_{\tau\mu}G_{\tau}\left(\lim\limits_{k\rightarrow \infty}x^k\right) - \lim\limits_{k\rightarrow \infty}x^k\right\|\\
&=\left\|S_{\tau\mu}G_{\tau}(x)-x^*\right\|.
\end{array}
\end{eqnarray*}
Hence
$x^*=\lim\limits_{k\rightarrow\infty}x^k$ is a fixed point of function $S_{\tau\mu}\circ G_{\tau}$, and we complete the proof.
\end{proof}

\section*{Acknowledgment}
This work was supported by Ministry of Science and Technology under grants MOST 102-2221-E-001-022-MY2 and 102-2221-E-001-002-MY2.

\ifCLASSOPTIONcaptionsoff
  \newpage
\fi



%
\bibliographystyle{IEEEtran}
\bibliography{refs}

%
%

%








\end{document}